\pdfoutput=1
\documentclass[10pt,pra,aps,superscriptaddress,nofootinbib,twocolumn]{revtex4-1}
\usepackage{subcaption}
\usepackage{amsmath,bbm}
\usepackage{amsthm}
\usepackage{amsmath}
\usepackage{latexsym}
\usepackage{amssymb}
\usepackage{float}
\usepackage{graphicx}           
\usepackage{color}
\usepackage{xcolor}
\usepackage{mathpazo}
\usepackage{comment}
\usepackage{enumitem}
\usepackage{multirow}
\usepackage{setspace}
\usepackage[colorlinks=true,linkcolor=blue,citecolor=magenta,urlcolor=blue]{hyperref}
\usepackage{hyperref}
\usepackage{svg}
\usepackage{adjustbox}
\usepackage{rotating}
\usepackage{caption}
\usepackage{subcaption}
\usepackage{soul}

\newcommand{\be}{\begin{equation}}
\newcommand{\ee}{\end{equation}}
\newcommand{\bea}{\begin{eqnarray}}
\newcommand{\eea}{\end{eqnarray}}

\def\squareforqed{\hbox{\rlap{$\sqcap$}$\sqcup$}}
\def\qed{\ifmmode\squareforqed\else{\unskip\nobreak\hfil
\penalty50\hskip1em\null\nobreak\hfil\squareforqed
\parfillskip=0pt\finalhyphendemerits=0\endgraf}\fi}
\def\endenv{\ifmmode\;\else{\unskip\nobreak\hfil
\penalty50\hskip1em\null\nobreak\hfil\;
\parfillskip=0pt\finalhyphendemerits=0\endgraf}\fi}
\DeclareMathOperator{\Tr}{Tr}

\newcommand{\rank}{\text{Rank}}

\newcommand{\ket}[1]{|#1\rangle}
\newcommand{\bra}[1]{\langle#1|}

\makeatletter
\newtheorem*{rep@theorem}{\rep@title}
\newcommand{\newreptheorem}[2]{%
\newenvironment{rep#1}[1]{%
 \def\rep@title{#2 \ref{##1}}%
 \begin{rep@theorem}}%
 {\end{rep@theorem}}}
\makeatother

\newtheorem{theorem}{Theorem}

\newtheorem{obs}{Observation}

\begin{document}


\title{Topology-independent quantum advantage in communication networks}


\author{Ankush Pandit}
\affiliation{Department of Physics, School of Basic Sciences, Indian Institute of Technology Bhubaneswar, Odisha 752051, India}
\author{V. N. S. Meghanath Ashtakala}
\affiliation{School of Physics, Indian Institute of Science Education and Research Thiruvananthapuram, Kerala 695551, India}
\affiliation{Department of Computer and Systems Sciences, Stockholm University, Sweden 10691}
\author{Leon George Padayatty}
\affiliation{School of Physics, Indian Institute of Science Education and Research Thiruvananthapuram, Kerala 695551, India}
\author{Debashis Saha}
\affiliation{Department of Physics, School of Basic Sciences, Indian Institute of Technology Bhubaneswar, Odisha 752051, India}


\begin{abstract}
Quantum communication can outperform classical communication in a variety of information-processing tasks, but such advantages are typically established for a fixed communication network. Here we introduce \emph{topology-independent quantum advantage}, in which a quantum protocol outperforms every classical protocol subject to the same communication constraint, irrespective of how the classical communication resource is distributed among the parties or which communication topology is employed, even when the parties may share unlimited classical randomness. We first demonstrate this phenomenon through a simple three-party equality task: two senders receive ternary classical inputs, while an input-free receiver determines whether the inputs are equal. We then formulate a general framework for topology-independent communication and completely characterize the corresponding classical correlation polytope in the minimal three-party scenario. Its facet inequalities reveal several instances of topology-independent quantum advantage. Remarkably, one of these inequalities enables the \textit{semi-device-independent certification} of a nonunitary quantum channel, while another certifies a sequential quantum network topology without any prior assumption about the network configuration. Finally, we extend the framework to four-party networks, characterize the relevant classical communication structures, and identify a communication task exhibiting quantum advantages against all admissible classical topologies. Our results establish topology independence as a genuine and robust form of quantum advantage in communication networks.
\end{abstract}

\maketitle


\section{Introduction} 

Communication complexity provides a fundamental framework for quantifying the physical resources required for distributed information processing. In its simplest form, a communication-complexity task involves two spatially separated parties, a sender and a receiver, who receive classical inputs and aim to compute a function of them under a constraint on the amount of communication. Quantum theory can reduce the communication resources required for such tasks: transmitting quantum systems can outperform protocols restricted to classical messages. Such advantages are particularly striking because they arise solely from the communication of low-dimensional quantum systems, without requiring preshared entanglement \cite{Raz,RevModPhys.82.665,DEWOLF,Acin-prl-2010,Tavakoli-prl-2015,ambainis2009,saha2023,saha-pra-2023,cruzeiro2026}.

Quantum advantages based solely on the communication of quantum systems have been extended beyond two-party prepare-and-measure protocols to multiparty communication networks \cite{QF,Zukowski2005,vertesi,Marcin2015,Hameedi_PRA,Saha_EPL,smania2016,Barnik-pra,chakraborty2024,Hameedi-pra-2020,Chaves-prr-2020,Pandit-pra-2026,sasmal2026}. In multiparty scenarios, quantum information may be transmitted through several parties, with intermediate parties performing input-dependent transformations before forwarding the system to subsequent nodes. These networks have also found application in semi-device-independent certification of quantum devices \cite{miklin2020semi,Mohan_self_test_2019,Roy_R_cert_2026,Mukherjee21,mitra2026}.

A common feature of essentially all such communication scenarios is that the \emph{communication topology}, specifying which party communicates with which other party, is fixed a priori. 
The classical and quantum resources are then compared within this prescribed architecture. 
Thus, while the communication resource has been extensively studied, the
\emph{communication topology itself} has so far largely remained part of the
assumed specification of the problem. Here we remove this assumption.

More precisely, suppose that only the total amount of communication is constrained, while a classical protocol is free to distribute this resource among the parties in an arbitrary way. The parties may choose who communicates with whom, how the available communication dimension is distributed among different links, and may even correlate this choice with unlimited shared classical randomness. A quantum protocol exhibiting an advantage against this entire class of classical strategies would therefore demonstrate a form of quantum advantage that is independent of the architecture of the competing classical network.

In this work, we introduce and systematically investigate precisely this notion, which we call \emph{topology-independent quantum advantage}. 
We first demonstrate the phenomenon through a remarkably simple three-party equality task. Alice and Bob independently receive ternary inputs, while Charlie receives no input and wants to determine whether their inputs are equal. The total communication is restricted to two bits. We show that communication of two qubits outperforms every classical strategy satisfying this constraint, irrespective of whether the classical messages are sent in parallel, sequentially, or through a randomized mixture of different communication topologies.

Motivated by this example, we formulate a general framework for topology-independent communication networks. We represent admissible communication structures by weighted sink-rooted directed acyclic graphs, with edge weights specifying the dimensions of the communicated messages. The topology-independent classical correlation set is then obtained as the convex hull of the correlations generated by all deterministic strategies over all admissible graphs, edge weights, party assignments, and causal orderings.

We completely characterize this classical set in the simplest nontrivial three-party scenario. Exploiting redundancies among communication topologies, we reduce the number of structures that need to be considered and obtain the facets of the resulting topology-independent classical polytope. Several of these facet inequalities admit quantum violations using qubit communication. Interestingly, the operational content of these inequalities goes beyond witnessing a communication advantage. For one facet, the maximum value achievable when the intermediate party is restricted to unitary transformations is strictly smaller than that attainable with a general quantum channel. Consequently, its violation beyond the unitary bound provides a semi-device-independent certification that the intermediate operation is necessarily nonunitary. Furthermore, we identify another facet inequality that is satisfied by all quantum correlations attainable with nonsequential communication topologies. Consequently, its violation provides a semi-device-independent certification that the underlying quantum network necessarily implements sequential communication, without requiring any prior knowledge of the network topology or the internal operation of the devices.

Finally, we extend the framework to a four-party network with three senders and one input-free receiver. Although the number of a priori possible communication graphs grows rapidly with the number of parties, structural reductions allow the relevant classical strategies to be represented by a substantially smaller family of topologies. We characterize the corresponding topology-independent classical correlations and identify a communication inequality displaying quantum violations. These results show that topology-independent quantum advantage is not restricted to the minimal three-party example, but constitutes a broader phenomenon in multipartite quantum networks.

\section{A simple example}\label{sec:min-exp}

We begin with a simple yet instructive example exhibiting a topology-independent quantum advantage. Consider a communication network consisting of two senders, Alice and Bob, and a receiver, Charlie. In each run of the protocol, Alice receives an input randomly from the set $x_1\in\{0,1,2\}$ and Bob receives an input randomly from the set $x_2\in\{0,1,2\}$. The parties are allowed to exchange at most two bits of communication in total. Importantly, these two bits may be distributed according to an arbitrary communication topology. For example, Alice may send one bit to Bob and Bob may subsequently send one bit to Charlie, or Alice and Bob may each send one bit directly to Charlie, or any other ways. Furthermore, the parties may share an arbitrary random variable and condition the communication topology on its value. Thus, the classical resources consist of shared randomness together with two bits of communication distributed over an arbitrary directed network.

The goal of Charlie is to guess the binary value function $f(x_1,x_2)=\delta_{x_1,x_2}$, that is, whether the inputs $x_1$ and $x_2$ are the same or not.
It can be verified that, with two bits of communication distributed according to an arbitrary communication topology, any classical strategy can produce the correct output for at most seven of the nine possible input pairs. A rigorous proof is given later.

In contrast, a quantum strategy employing two communicated qubits achieves better success rate in producing the correct answer. Consider the protocol in which Alice encodes her input into a qubit state and sends it to Bob. The encoding is given by $\rho_{x_1} = \ket{\psi_{x_1}}\!\bra{\psi_{x_1}}$, where
\begin{equation}
\ket{\psi_{x_1}} = \frac{1}{\sqrt{2}}\left(\ket{0} + \exp^{-i 2\pi x_1/3}\ket{1}\right), \ x_1=0,1,2,
\label{eq:alice_encoding}
\end{equation}
with $\{\ket{0},\ket{1}\}$ denoting the computational basis.
Upon receiving the qubit, Bob applies a quantum channel $\mathcal{E}_{x_2}$ depending on his input $x_2$. Each channel is specified by two Kraus operators,
$\mathcal{E}_{x_2}(\rho)=K_{x_2}\rho K_{x_2}^\dagger+K_{x_2}'\rho K_{x_2}'^\dagger$.
For $x_2=0$,
\begin{equation}
K_0 = \begin{pmatrix}
    \frac{-1-i}{2} & 0 \\
    0 & \frac{1}{\sqrt{2}}
\end{pmatrix}, \quad 
K_0' = \begin{pmatrix}
    0 & \frac{1}{\sqrt{2}} \\
    \frac{-1+i}{2} & 0
\end{pmatrix}.
\end{equation}
For $x_2=1$,
\begin{equation}
K_1 = \begin{pmatrix}
    -\frac{\sqrt{3}-1}{4} + \frac{\sqrt{3}+1}{4}i & 0 \\
    0 & \frac{1}{\sqrt{2}}
\end{pmatrix}, 
K_1' = \begin{pmatrix}
    0 & \frac{1}{\sqrt{2}} \\
    \frac{\sqrt{3}+1}{4} + \frac{\sqrt{3}-1}{4}i & 0
\end{pmatrix}
\end{equation}
For $x_2=2$,
\begin{equation}
K_2 = \begin{pmatrix}
    \frac{\sqrt{3}+1}{4} - \frac{\sqrt{3}-1}{4}i & 0 \\
    0 & \frac{1}{\sqrt{2}}
\end{pmatrix},
K_2' = \begin{pmatrix}
    0 & \frac{1}{\sqrt{2}} \\
   - \frac{\sqrt{3}-1}{4} - \frac{\sqrt{3}+1}{4}i & 0
\end{pmatrix}.
\end{equation}
Bob then forwards the resulting qubit to Charlie. Charlie performs the measurement in the basis $ \left\{
\ket{\phi}\!\bra{\phi},
\; \ket{\phi^\perp}\!\bra{\phi^\perp}
\right\},$
where 
\be 
\ket{\phi} = -\frac{1+i}{2}\ket{0}
+\frac{1}{\sqrt{2}}\ket{1}.
\ee 
Charlie outputs $0$ whenever the outcome corresponding to
$\ket{\phi}$ occurs, and outputs $1$ otherwise.

A direct calculation shows that the resulting conditional probabilities are given by,
\begin{eqnarray}
    & p(0|x_1,x_2) = 1, \text{ whenever } x_1=x_2; \nonumber \\ 
    & p(1|x_1,x_2) = \frac34, \text{ whenever } x_1\neq x_2.
\end{eqnarray}
The average success rate of this task is witnessed by the following expression,
\begin{equation}
\label{eq:I-1}
\mathcal{I}= \frac{1}{9} \sum_{x_1,x_2}
p\!\left(\delta_{x_1,x_2}\mid x_1,x_2\right).
\end{equation}
Every classical protocol employing at most two communicated bits, irrespective of the communication topology, satisfies $\mathcal{I}\le 7/9$, whereas the quantum protocol described above attains the algebraic maximum of $\mathcal{I}=5/6$. This immediately leads to the following observation.
\begin{obs}
There exists a bounded-dimensional communication task in a topology-independent communication network consisting of three parties, in which the output party receives `no input', such that quantum communication outperforms classical.
\end{obs}

\section{Topology-independent communication network}
We now introduce the general framework of topology-independent communication networks. Throughout the paper, we denote by
$[K]=\{0,1,\ldots,K-1\}$ the set of integers modulo $K$. Consider a network consisting of $N$ sender parties $S_1,\ldots,S_{N}$ and one receiver $R$. In each run of the protocol, every sender $S_i$ receives an input
$x_i\in[n_{x_i}]$, while the receiver produces an output $z\in[n_z]$. We emphasize that, throughout this work, the receiver is assumed not to receive an external input. Although the framework can be straightforwardly extended to include an input at the receiver, none of the quantum advantages presented here require such a scenario. The communication task is completely specified by the conditional probability distribution $p(z|\vec{x})$, where
$\vec{x}=(x_1,\ldots,x_{N})$ collects the inputs of all sender parties.

The defining feature of the present framework is that the communication topology is \emph{not fixed}. Instead, the communicating parties are allowed to distribute the communication resources according to an arbitrary topology. Consequently, every protocol compatible with the prescribed communication resources is regarded as admissible. We therefore refer to this scenario as a \emph{topology-independent communication network}. The only resource constraint is the total dimension of the communicated messages. We denote the corresponding communication scenario by $[\vec{n}_x,n_z,d]$, where
$\vec{n}_x=(n_{x_1},\ldots,n_{x_{N}})$ specifies the input alphabet sizes of the sender parties.

A communication topology is naturally represented by a directed acyclic graph (DAG). Each vertex corresponds to one party, while a directed edge $(i,j)$ represents communication from party $i$ to party $j$. Since communication follows a causal order, directed cycles are forbidden. Every edge is assigned a positive integer weight $d_{(i,j)}\ge2$, representing the dimension of the communicated message. An absent edge is interpreted as the absence of communication.

Since the receiver is distinguished from the sender parties, every relevant communication topology can be represented as a \textit{weighted sink-rooted DAG}, whose unique sink corresponds to $R$. The edge weights satisfy $d_{(i,j)}\ge2$, together with the resource constraint $\sum_{(i,j)} d_{(i,j)}=d.$

Not every sender needs to participate in the communication. If a sender is not connected to the receiver by any directed path, then its communicated messages cannot influence the output distribution $p(z|\vec{x})$. Consequently, without loss of generality, such disconnected components may be discarded. Therefore, every relevant communication topology is represented by a \textit{weighted sink-rooted DAG} containing between two and $N+1$ vertices.

Besides the graph itself, a communication protocol also requires a causal ordering of the communication events. Any ordering consistent with the partial order induced by the directed graph is admissible. Different linear extensions of this partial order generally correspond to distinct communication protocols.

\begin{obs}
The complete set of communication topologies compatible with the scenario
$[\vec{n}_x,n_z,d]$ is obtained by considering
\begin{enumerate}
\item all weighted sink-rooted directed acyclic graphs containing between two and $N+1$ vertices;
\item all assignments of the sender parties to the non-root vertices, i.e., all permutations of the labels $\{S_i\}_{i=1}^{N}$;
\item all feasible edge-weight assignments satisfying $d_{(i,j)}\ge2$ and $\sum_{(i,j)} d_{(i,j)}=d.$
\item all causal orderings consistent with the partial order defined by the directed graph.
\end{enumerate}
\end{obs}

A particularly important subclass consists of sink-rooted trees, namely DAGs containing exactly $N$ edges and $N+1$ vertices. In this case, the causal ordering is unique up to the ordering of vertices having the same depth. Consequently, it is sufficient to consider the natural ordering in which communication proceeds from larger depth towards the root.

The topology-independent classical correlation set is obtained by considering every deterministic communication strategy compatible with every admissible communication topology. Each deterministic strategy defines a vertex of the correlation polytope through the conditional probabilities
$p(z|\vec{x})$. The convex hull of all such vertices over all admissible topologies therefore defines the topology-independent classical communication polytope. By construction, this polytope contains the classical correlation polytope associated with every fixed communication topology as a convex subset.

\section{An explicit study of the three-party scenario} \label{sec:3party}
Since the receiver has no input and the parties are allowed to share arbitrary classical randomness, the Frenkel--Wehner theorem~\cite{fw} rules out a quantum advantage in the one-sender--one-receiver scenario. A topology-independent quantum advantage therefore requires at least two senders and one receiver. Moreover, in a three-party network the smallest nontrivial communication resource corresponds to two communicated bits, or equivalently total communication dimension $d=4$. For $d<4$, at most one nontrivial binary message can be communicated, effectively reducing the problem to a one-sender--one-receiver scenario.

The scenarios $[(2,2),n_z,4]$ and $[(2,3),n_z,4]$ likewise cannot exhibit a topology-independent quantum advantage. In the former, the two senders can communicate their complete inputs to the receiver using one bit each. In the latter, consider the sequential topology $A\rightarrow B\rightarrow C$: Alice can communicate her complete binary input to Bob, who can then regard $(x_1,x_2)$ as a single local input before communicating with Charlie. The resulting protocol is therefore effectively a one-sender--one-receiver communication problem, for which the Frenkel--Wehner theorem~\cite{fw} precludes a quantum advantage. This leads to the following theorem.
\begin{obs}
The simplest scenario admitting a topology-independent quantum advantage is $[(3,3),2,4]$.
\end{obs}

Before proceeding, we establish two useful theorems that substantially simplifies the characterization of two-sender communication scenarios. The proofs of these two theorems are provided in Appendix~\ref{App:proofs}.
\begin{theorem}\label{lem:1}
Consider the topology $A\rightarrow C,\;B\rightarrow C$, where Alice and Bob communicate messages of dimensions $d_A$ and $d_B$,
respectively, and Charlie produces an output $z\in[n_z]$.
Every classical or quantum correlation achievable in this topology can be reproduced in the sequential topology
$A\rightarrow B\rightarrow C$ using communication dimensions
$d_A$ and $n_z$ on the two respective edges. Analogously, the same
correlation can be reproduced in $B\rightarrow A\rightarrow C$ using
dimensions $d_B$ and $n_z$.
\end{theorem}

\begin{figure}[htbp]
    \centering
    \includegraphics[width=0.9\linewidth]{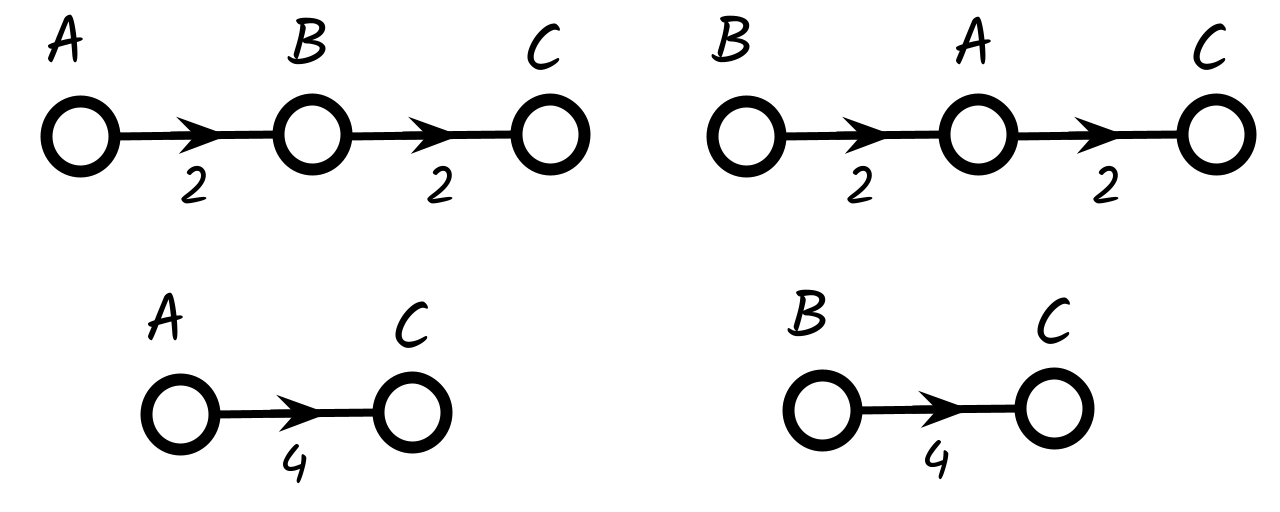}

    \caption{All communication topologies sufficient to characterize the topology-independent classical correlations in the $[(3,3),2,4]$ scenario. The nodes $A$, $B$, and $C$ represent the parties, while directed edges indicate the allowed communication channels, with the corresponding message dimensions indicated alongside the edges.}\label{fig:1}
    \label{fig:placeholder}
\end{figure}

\begin{theorem}\label{lem:2}
    For the scenario $[(3,3),2,4]$, it is sufficient to consider only four topologies, shown in Fig. \ref{fig:1}. 
\end{theorem}

For a fixed topology, the most general classical strategy of each party is specified by a conditional probability distribution for its outgoing message, conditioned on its local input, all previously received messages, and the shared random variable. Allowing the parties to correlate both their strategies and the choice of communication topology through shared randomness, the general classical correlation can be written as
\begin{eqnarray}
\label{pc3}
p_{\mathcal{C}}(z|x_1,x_2)
&=&
\sum_{\lambda_1} q_{\lambda_1} \,
p_{\lambda_1}^{A\rightarrow B\rightarrow C}
(z|x_1,x_2,\lambda_1)
 \nonumber \\ 
&& + \sum_{\lambda_2} q_{\lambda_2} \,
p_{\lambda_2}^{B\rightarrow A\rightarrow C}
(z|x_1,x_2,\lambda_2)
\nonumber\\
&&+
\sum_{\lambda_3} q_{\lambda_3} \,
p_{\lambda_3}^{A\rightarrow C}
(z|x_1,x_2,\lambda_3) \nonumber \\
&& + 
\sum_{\lambda_4} q_{\lambda_4} \,
p_{\lambda_4}^{B\rightarrow C}
(z|x_1,x_2,\lambda_4),
\end{eqnarray}
where
$q_{\lambda_i}>0,
\ 
\sum_i\sum_{\lambda_i}q_{\lambda_i}=1.$
For the two sequential topologies, in which one bit is communicated along each edge,
\begin{eqnarray}
&& p_{\lambda_1}^{A\rightarrow B\rightarrow C}
(z|x_1,x_2,\lambda_1) \nonumber \\
&=&
\sum_{m_A,m_B\in\{0,1\}}
p(m_A|x_1,\lambda_1)
p(m_B|m_A,x_2,\lambda_1)
p(z|m_B,\lambda_1) \nonumber \\
\end{eqnarray}
and
\begin{eqnarray}
&& p_{\lambda_3}^{B\rightarrow A\rightarrow C}
(z|x_1,x_2,\lambda_2) \nonumber \\
&=&
\sum_{m_B,m_A\in\{0,1\}}
p(m_B|x_2,\lambda_2)
p(m_A|m_B,x_1,\lambda_2)
p(z|m_A,\lambda_2). \nonumber \\
\end{eqnarray}
When a single sender uses the entire two-bit communication resource to communicate directly with Charlie, while the other sender does not participate, the corresponding correlations are
\begin{equation}
p_{\lambda_4}^{A\rightarrow C}
(z|x_1,x_2,\lambda_4)
=
\sum_{m\in\{0,1,2,3\}}
p(m|x_1,\lambda_4)
p(z|m,\lambda_4),
\end{equation}
and
\begin{equation}
p_{\lambda_5}^{B\rightarrow C}
(z|x_1,x_2,\lambda_5)
=
\sum_{m\in\{0,1,2,3\}}
p(m|x_2,\lambda_5)
p(z|m,\lambda_5).
\end{equation}

The scenario contains $18$ conditional probabilities
$\{p(z|x_1,x_2)\}$. Since Charlie's output is binary, normalization allows us to work in the nine-dimensional representation
$p_{x_1,x_2}:=p(0|x_1,x_2)$,
with $p(1|x_1,x_2)=1-p(0|x_1,x_2).$ Thus, any linear figure of merit can be expressed entirely in terms of the nine probabilities $\{p_{x_1,x_2}\}$.

For each fixed topology, it is sufficient to enumerate deterministic communication strategies, since arbitrary stochastic strategies are obtained as convex combinations thereof. The topologies $A\rightarrow C,\;B\rightarrow C$, $A\rightarrow C$, and $B\rightarrow C$ each generate $1024$ deterministic strategies for $\{p(z|x_1,x_2)\}$, whereas the sequential topologies $A\rightarrow B\rightarrow C$ and $B\rightarrow A\rightarrow C$ each generate $2048$.
The topology-independent classical correlation polytope is obtained by taking the convex hull of the deterministic correlations generated by all five topologies in Eq.~\eqref{pc3}. Since different deterministic strategies, including strategies associated with different topologies, may generate identical correlations, duplicate points are removed before constructing the convex hull. We then determine the facets of the resulting polytope using \texttt{SageMath / PANDA}. This procedure yields $330$ facet inequalities, from which trivial constraints, such as
$1\geq p(z|x_1,x_2)\geq 0,$
are discarded.

To identify inequivalent facets, we quotient the resulting inequalities by the natural relabeling symmetries of the scenario: (i) permutations of Alice's input labels $x_1$; (ii) permutations of Bob's input labels $x_2$; (iii) exchange of Alice and Bob; and (iv) exchange of Charlie's output labels, $z:0\leftrightarrow1$. Under these symmetry operations, the nontrivial facets reduce to $11$ inequivalent classes, listed in Table~\ref{TABLE1}.

\begin{widetext}
\begin{center}
\begin{table}[h!]
\begin{tabular}{|c|c|c|c|}
\hline
\hline
S.No. &  Inequalities  & ~~$\mathcal{S_Q}^{\text{LB}}$ ~~& $v^*$  \\
\hline
\hline
1   & $\mathcal{I}_1 = p_{00} -  p_{01}  -p_{02} -  p_{10}  +p_{11} -p_{12} -p_{20} -p_{21} +  p_{22} \leqslant 1$   & 1.5& 0.101 \\ 
\hline

2  &  $\mathcal{I}_2 = -p_{00}   -p_{02}  +p_{10} -  p_{11} -p_{12}   -p_{21} +  p_{22} \leqslant 1$  & 1&  0 \\ 
\hline

3   &  $\mathcal{I}_3 = -p_{00} -  p_{01}  -p_{02}  +p_{10} -2  p_{11}  +2p_{12} -  2p_{20}  +p_{21} +  2p_{22} \leqslant 4$  & 4.182 &  0.023 \\ 
\hline

4    &  $\mathcal{I}_4 = p_{00}  +p_{01} +p_{02}+2p_{10} - 2p_{11}  -p_{12}  -p_{20} -2p_{21} + 2 p_{22} \leqslant 5$ &5.182 &  0.021  \\ 
\hline

5  &$ \mathcal{I}_5 =   -p_{01}  +p_{02}  -p_{10} +  p_{11}  +  p_{20}  -p_{22} \leqslant 2$    & 2.5981&   0.123   \\ 
\hline

6  &  $ \mathcal{I}_6 = 2p_{01} -p_{02}   -p_{10} - p_{11} +p_{12}  +2p_{20}- 2p_{21}  -p_{22} \leqslant 3$    & 3.66&   0.086 \\ 
\hline

7  & $\mathcal{I}_7 = -p_{00} +  p_{01}  -p_{02}  +  p_{12}  + p_{20}  -p_{22} \leqslant  2$  & 2.414&  0.095 \\ 
\hline

8  & $\mathcal{I}_8 = p_{00} +p_{01} +p_{10} -p_{11} +p_{12} -p_{20} +p_{22} \leqslant  4$  &4.414 &  0.083    \\ 
\hline

9   & $\mathcal{I}_9 = p_{00} -  p_{01}  + p_{02}   +p_{11} - 2p_{12}  -2p_{20} +  p_{21} +  2p_{22} \leqslant 4$  & 4.66& 0.077 \\ 
\hline

10  & $ \mathcal{I}_{10} = -p_{00}  +  p_{10}  -  p_{12} +  p_{20} -  p_{21}  +p_{22} \leqslant  2$  & 2.414&  0.093  \\  \hline

11  & $ \mathcal{I}_{11} = p_{00}  +p_{01} -p_{02}  - p_{10}  +p_{11} +  p_{12}  +p_{20} - p_{21}  + p_{22} \leqslant 4$  & 4.5&  0.077 \\ 

\hline \hline












\end{tabular}
\caption{The table lists the complete set of inequivalent facet inequalities characterizing the $[(3,3),2,4]$ topology-independent classical correlations. For the sequential topology $A\rightarrow B\rightarrow C$, the reported quantum values, $\mathcal{S_Q}^{\mathrm{LB}}$, are lower bounds obtained through a see-saw optimization over Alice's preparations, Bob's quantum channels, and Charlie's measurements. The robustness of each violation is quantified using a white-noise model in which the communicated qubit is replaced according to $\rho\rightarrow(1-v)\rho+v\mathbb{I}/2$ at both communication stages, $A\rightarrow B$ and $B\rightarrow C$. Here, $v^{*}$ denotes the maximum tolerable white-noise fraction for which the corresponding classical bound remains violated, or, equivalently, $1-v^{*}$ is the minimum visibility required to observe the violation. $\mathcal{I}_1$ is equivalent to the communication task defined in Eq.~\eqref{eq:I-1} of Sec.~\ref{sec:min-exp}. Explicit quantum strategies attaining the reported $\mathcal{S}_{Q}^{\mathrm{LB}}$ for $\mathcal{I}_5$ and $\mathcal{I}_8$ are provided in Appendix~\ref{app:Q-strategies}. Note that, see-saw optimization for the parallel topology $A\rightarrow C\leftarrow B$ yields no quantum violation of any of the listed inequalities.}
\label{TABLE1}
 \end{table}
 \end{center}
\end{widetext}

Owing to Theorem~\ref{lem:1}, it suffices to consider the sequential topology
$A\rightarrow B\rightarrow C$ when searching for a quantum advantage. We optimize the corresponding quantum strategies using semi-definite programming (SDP) in a seesaw procedure. A strategy consists of Alice's preparations $\{\rho_x\}$, Bob's input-dependent quantum channels $\{\mathcal{N}_y\}$, and Charlie's POVM $\{M_z\}$. At each step, two of these three sets of variables are held fixed while the third is optimized with respect to the linear functional defined by a given facet inequality. The three optimizations are iterated until convergence, and the largest values obtained over multiple initializations are reported in Table~\ref{TABLE1}.

To formulate the channel optimization as an SDP, we employ the Choi--Jamio{\l}kowski representation~\cite{miklin2020semi}. For Bob's channel $\mathcal{N}_y$, let
\begin{equation}
\Phi_y
=
(\mathbb{I}\otimes\mathcal{N}_y)
\left(
|\Omega\rangle\langle\Omega|
\right),
\end{equation}
where $|\Omega\rangle=\sum_i |ii\rangle $ is the unnormalized maximally entangled state. The action of the channel is then
$\mathcal{N}_y(\rho_x)
=
\operatorname{Tr}_{\mathrm{in}}
\left[
(\rho_x^{T}\otimes\mathbb{I})\Phi_y
\right].$ Consequently, the observed probabilities admit the compact representation
\begin{equation}
\label{eq5}
p(z|x,y)
=
\operatorname{Tr}
\left[
(\rho_x^{T}\otimes M_z)\Phi_y
\right].
\end{equation}
A matrix $\Phi_y$ represents a completely positive trace-preserving quantum channel if and only if
$\Phi_y\succeq0, 
\operatorname{Tr}_{\mathrm{out}}\Phi_y=\mathbb{I}_{\mathrm{in}}.$ Together with $\rho_x\succeq0, 
\operatorname{Tr}(\rho_x)=1,$
and $M_z\succeq0, 
\sum_z M_z=\mathbb{I},$
these conditions render each individual step of the seesaw optimization an SDP. The resulting quantum values, together with the robustness of the corresponding violations against white noise, are reported in Table~\ref{TABLE1}.

\subsection{Semi-device independent certification of non-unitary channels} 

We next show that some facet inequalities can certify the nonunitarity of Bob's channel in a semi-device-independent manner. To this end, we derive an upper bound on each facet functional over all quantum strategies in the topology $A\rightarrow B\rightarrow C$ under the restriction that Bob applies only unitary channels. We denote this bound by $\mathcal{I}_{U}^{A\rightarrow B\rightarrow C}$.

In the $[(3,3),2,4]$ scenario, Charlie has a binary output measurement. Since the figure of merit is linear in the observed probabilities, the optimization over Charlie's measurement can be restricted to extremal binary qubit POVMs, which, in the nontrivial case, are projective. Let $\{\ket{\phi_0},\ket{\phi_1}\}$ denote the corresponding measurement basis. If Bob's channel for input $x_2$ is unitary, the observed probabilities take the form
\begin{equation}
\label{ptm-pm}
p(z|x_1,x_2)
=
\bra{\phi_z}
U_{x_2}\rho_{x_1}U_{x_2}^{\dagger}
\ket{\phi_z}
=
\bra{\phi_{z|x_2}}
\rho_{x_1}
\ket{\phi_{z|x_2}},
\end{equation}
where
$\ket{\phi_{z|x_2}} =
U_{x_2}^{\dagger}\ket{\phi_z}.$
Hence, every correlation obtainable with unitary channels in the sequential topology admits a qubit prepare-and-measure realization: Alice prepares $\rho_{x_1}$, while the measurement choice $x_2$ specifies the projective measurement
$\{\ket{\phi_{0|x_2}},\ket{\phi_{1|x_2}}\}$.

Consequently, the maximum of any linear functional of
$\{p(z|x_1,x_2)\}$ over unitary sequential strategies is upper bounded by its maximum in the corresponding dimension-bounded prepare-and-measure scenario. We bound the latter using the Navascu\'es--V\'ertesi semidefinite-programming hierarchy~\cite{PhysRevLett.115.020501}. The resulting bounds
$\mathcal{I}_{U}^{A\rightarrow B\rightarrow C}$ for all inequivalent facet inequalities are reported in Table~\ref{TABLE1}.

For one facet, this unitary-channel bound is strictly smaller than the quantum value attainable with general channels. Therefore, observing a value above
$\mathcal{I}_{U}^{A\rightarrow B\rightarrow C}$ rules out every unitary realization of Bob's operation. Importantly, this conclusion does not rely on a characterization of the preparations, channels, or measurements, but only on the assumed dimensions of the communicated systems. It therefore provides a semi-device-independent certification of a nonunitary quantum channel. To the best of our knowledge, this constitutes one of the simplest such certifications in terms of the number of preparations and measurement settings.
\begin{obs}
In the $[(3,3),2,4]$ scenario, an observed violation of
\begin{equation}\label{I8u}
\mathcal{I}_8 = p_{00}+p_{01}+p_{10}-p_{11}+p_{12}-p_{20}+p_{22}
\leq 4.0981
\end{equation}
certifies, under the assumed communication-dimension constraint, that the intermediate party implements a nonunitary quantum channel. There exist quantum strategies that violate \eqref{I8u}. 
\end{obs}

\subsection{Semi-device-independent certification of quantum network topology}
\label{sec:topology-certification}

We now show that topology-independent communication inequalities have a
further application: they can be used to certify the topology of a quantum
communication network without assuming the topology \emph{a priori}.
Specifically, we assume only the input and output alphabets and an upper
bound on the total communication dimension, while the communicated states,
quantum operations, and measurements are otherwise uncharacterized. The
result therefore provides a semi-device-independent certification of the
underlying quantum network topology.

Consider the scenario $[(3,3),2,4]$ and inequality~5 of
Table~\ref{TABLE1},
\begin{equation}
    \mathcal{I}_5 =
    -p_{01}+p_{02}-p_{10}+p_{11}+p_{20}-p_{22}
    \leq 2 ,
    \label{eq:I5-topology}
\end{equation}
where $p_{xy}\equiv p(0|x,y)$. As shown above, sequential quantum
communication can violate this inequality. We next show that such a
violation is impossible for the parallel topology, as well as for either
direct topology, under the same total communication constraint.

\begin{theorem}\label{lem:parallelQ}
For the parallel quantum topology $A\rightarrow C \leftarrow B,$ where Alice and Bob independently communicate qubits to Charlie, and for the direct topologies $A\rightarrow C$ and $B\rightarrow C$ with a four-dimensional communicated system, the
quantum values of Eq.~\eqref{eq:I5-topology} satisfy
\begin{equation}
\mathcal{I}_{5,A\rightarrow C \leftarrow B}^{Q}\leq 2, \quad \mathcal{I}_{5,A\rightarrow C}^{Q}\leq2,
    \quad   \mathcal{I}_{5,B\rightarrow C}^{Q}\leq2.
    \label{eq:I5-par-bound}
\end{equation}
\end{theorem}
The proof is provided in Appendix~\ref{App:proofs}. The theorem provides a semi-device-independent witness of the communication
topology. Under the assumption of total communication dimension $d=4$, an observed value $\mathcal{I}_5>2$ rules out the parallel topology $A\rightarrow C\leftarrow B$ as well as
the direct topologies $A\rightarrow C$ and $B\rightarrow C$. Within the
set of admissible topologies considered here, the observed correlations
must therefore originate from a sequential quantum network, $A\rightarrow B\rightarrow C$ or
$B\rightarrow A\rightarrow C.$ Importantly, this conclusion requires neither a characterization of the
communicated states nor knowledge of the intermediate operations or
Charlie's measurement. Thus, the same communication inequality that
witnesses topology-independent quantum advantage simultaneously certifies
a genuinely sequential structure of the quantum communication network.

We finally note that this phenomenon does not appear to be restricted to
inequality~\eqref{eq:I5-topology}. Numerical seesaw optimization over
parallel qubit strategies returns, for each of the eleven inequivalent
inequalities in Table~\ref{TABLE1}, a maximal value equal, within
numerical precision, to its topology-independent classical bound. This
provides numerical evidence that other inequalities identified here may
also serve as semi-device-independent witnesses of sequential quantum
network topology. Establishing rigorous parallel quantum bounds for these
inequalities is an interesting direction for future work.

\section{Four-party scenario}

We now turn to the four-party scenario $[(3,3,3),2,6]$, consisting of three senders $\{S_i\}_{i=1}^{3}$ and a receiver $R$. Each sender receives one of three possible inputs, Charlie has a binary output, and the total communication dimension is bounded by $d=6$. The increased number of parties gives rise to a substantially richer set of communication topologies than in the three-party case.

\begin{theorem}
\label{lem:four-party-reduction}
In the topology-independent scenario $[(3,3,3),2,6]$, it is sufficient to consider only two classes of weighted sink-rooted directed acyclic graphs shown in Fig.~\ref{Fig2}. Including all relevant party permutations, these generate only $18$ distinct communication topologies.
\end{theorem}

\begin{figure}[h!]
    \centering
     \begin{subfigure}{0.19\textwidth}
         \centering
         \includegraphics[scale=0.26]{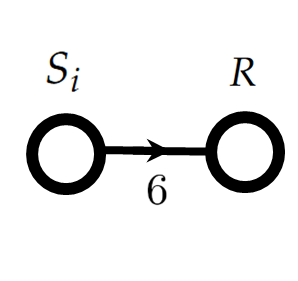}
         \caption{}\label{fig4:1}
     \end{subfigure}
     \hfill
     \begin{subfigure}{0.19\textwidth}
         \centering
    \includegraphics[scale=0.8]{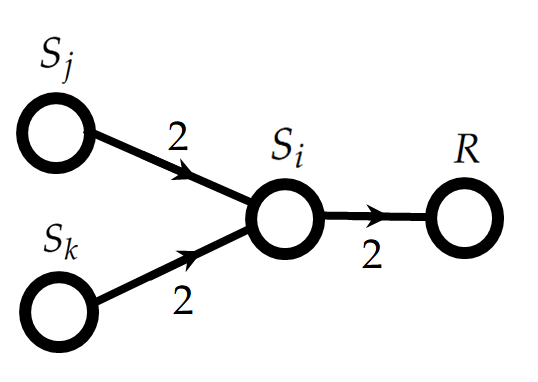}
         \caption{}\label{fig4:2}
     \end{subfigure}
     \hfill 
     \begin{subfigure}{0.3\textwidth}
         \centering
        \includegraphics[scale=0.26]{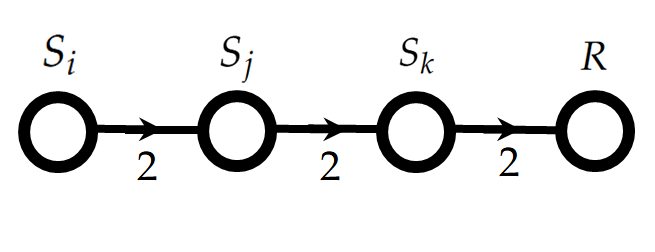}
         \caption{}\label{fig4:3a}
     \end{subfigure}
          \hfill
     \begin{subfigure}{0.3\textwidth}
         \centering
         \includegraphics[scale=0.26]{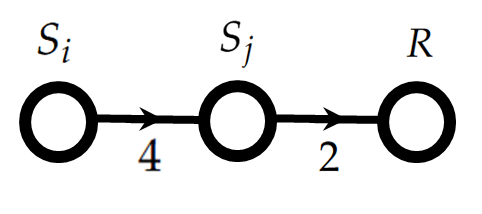}
         \caption{}
         \label{fig4:3b}
     \end{subfigure}
     
     \caption{All communication topologies sufficient to characterize the topology-independent classical correlations in the $[(3,3,3),2,6]$ scenario. The directed edges indicate the allowed communication channels, with the corresponding message dimensions indicated alongside the edges. (a) Here, $i\in\{1,2,3\}$, resulting in 3 topologies. (b) Here, $(i, j)$ represents all possible combinations from the set $\{1, 2, 3\}$, resulting in 6 distinct topologies. (c) Here, $(i,j,k)$ represents all possible permutations from the set $\{1,2,3\}$, resulting in 6 topologies. (d) Here, $i\in \{1,2,3\}$, with $j$ and $k$ being distinct from $i$ and each other, resulting in 3 topologies.}
     \label{Fig2}
 \end{figure}

As a consequence of Theorem~\ref{lem:four-party-reduction}, the most general topology-independent classical correlation in the $[(3,3,3),2,6]$ scenario can be written as a convex mixture of correlations generated by the four classes of DAGs identified above. Explicitly,
\begin{eqnarray}
\label{eq:classical-four-party}
p_{\mathcal C}(z|\vec{x})
&=&
\sum_{i=1}^{3}q_{\lambda_i}
p^{S_i\rightarrow R}(z|\vec{x},\lambda_i)
\nonumber\\
&&+
\sum_{\substack{i,j\in\{1,2,3\}\\i\neq j}}
q_{\mu_{i,j}}
p^{S_i\rightarrow S_j\rightarrow R}
(z|\vec{x},\mu_{i,j})
\nonumber\\
&&+
\sum_{\substack{i,j,k\in\{1,2,3\}\\i\neq j,\;j\neq k,\;k\neq i}}
q_{\nu_{i,j,k}}
p^{S_i\rightarrow S_j\rightarrow S_k\rightarrow R}
(z|\vec{x},\nu_{i,j,k})
\nonumber\\
&&+
\sum_{i=1}^{3}q_{\kappa_i}
p^{\substack{
S_{i\oplus_3 1}\rightarrow S_i,\,
S_{i\oplus_3 2}\rightarrow S_i,\\
S_i\rightarrow R}}
(z|\vec{x},\kappa_i),
\end{eqnarray}
where $\vec{x}=(x_1,x_2,x_3)$, $\oplus_3$ denotes addition modulo $3$ with the sender labels understood cyclically, and all mixing weights are non-negative and satisfy
\begin{equation}
\sum_{i=1}^{3}q_{\lambda_i}
+
\sum_{\substack{i,j\in\{1,2,3\}\\i\neq j}}q_{\mu_{i,j}}
+
\sum_{\substack{i,j,k\in\{1,2,3\}\\i\neq j,\;j\neq k,\;k\neq i}}
q_{\nu_{i,j,k}}
+
\sum_{i=1}^{3}q_{\kappa_i}
=1.
\end{equation}
Each of the four terms in \eqref{eq:classical-four-party} is defined as follows.
For the single-edge topology, one sender may use the entire communication resource and transmit a message of dimension at most $6$. The corresponding correlations are
\begin{equation}
p^{S_i\rightarrow R}(z|\vec{x},\lambda_i)
=
\sum_{m_i=0}^{5}
p(m_i|x_i,\lambda_i)
p(z|m_i,\lambda_i).
\end{equation}
For the two-edge sequential topology, it is sufficient to assign dimensions $4$ and $2$ to the first and second edges, respectively. Thus,
\begin{equation}
\begin{split}
p^{S_i\rightarrow S_j\rightarrow R}
(z|\vec{x},\mu_{i,j})
={}&
\sum_{\substack{m_i\in\{0,1,2,3\}\\m_j\in\{0,1\}}}
p(m_i|x_i,\mu_{i,j})\\
&\times
p(m_j|x_j,m_i,\mu_{i,j})
p(z|m_j,\mu_{i,j}).
\end{split}
\end{equation}
For the fully sequential three-edge topology, each communicated message is binary, and
\begin{equation}
\begin{split}
&p^{S_i\rightarrow S_j\rightarrow S_k\rightarrow R}
(z|\vec{x},\nu_{i,j,k})
\\
&=
\sum_{m_i,m_j,m_k\in\{0,1\}}
p(m_i|x_i,\nu_{i,j,k})
p(m_j|x_j,m_i,\nu_{i,j,k})
\\
&\hspace{20mm}\times
p(m_k|x_k,m_j,\nu_{i,j,k})
p(z|m_k,\nu_{i,j,k}).
\end{split}
\end{equation}
Finally, for the converging topology, let $S_i$ denote the sender adjacent to $R$ and let $S_j$ and $S_k$ be the remaining two senders. Its correlations take the form
\begin{equation}
\begin{split}
&p^{\substack{
S_j\rightarrow S_i,\,
S_k\rightarrow S_i,\\
S_i\rightarrow R}}
(z|\vec{x},\kappa_i)
\\
&=
\sum_{m_i,m_j,m_k\in\{0,1\}}
p(m_j|x_j,\kappa_i)
p(m_k|x_k,\kappa_i)
\\
&\hspace{15mm}\times
p(m_i|x_i,m_j,m_k,\kappa_i)
p(z|m_i,\kappa_i).
\end{split}
\end{equation}
Here $i,j,k$ are mutually distinct. Arbitrary local stochasticity can be absorbed into the shared random variable, so the extremal points of each correlation set are obtained by enumerating deterministic response functions. The resulting numbers of distinct vertices for the four classes of DAGs are summarized in Table~\ref{tab:four-party-vertices}.

\begin{table}[h!]
\centering
\begin{tabular}{l|c|c}
\hline\hline
DAG class & No.\ of & No.\ of \\
 & topologies & vertices \\
\hline
$S_i\rightarrow R$
    & $3$ & $13824$\\[1mm]
$S_i\rightarrow S_j\rightarrow R$
    & $6$ & $1048576$ \\[1mm]
$S_i\rightarrow S_j\rightarrow S_k\rightarrow R$
    & $6$ & $131072$ \\[1mm]
$S_j\rightarrow S_i,\;S_k\rightarrow S_i,\;S_i\rightarrow R$
    & $3$ & $4096$ \\
\hline\hline
\end{tabular}
\caption{Number of distinct vertices generated by each class of DAGs required to characterize topology-independent classical correlations in the $[(3,3,3),2,6]$ scenario.}
\label{tab:four-party-vertices}
\end{table}

Since $z$ is binary, normalization allows us to represent the correlations by the $27$ probabilities
$p_{x_1x_2x_3}:=p(0|x_1,x_2,x_3),$. The topology-independent classical set is the convex hull of the vertices generated by all $18$ topologies above. In contrast to the three-party scenario, a complete facet enumeration of this high-dimensional polytope is computationally demanding. We therefore search for nontrivial supporting inequalities by considering linear functionals
\begin{equation}
\mathcal I
=
\sum_{x_1,x_2,x_3}
\alpha_{x_1x_2x_3}
p_{x_1x_2x_3},
\ \ 
\alpha_{x_1x_2x_3}\in\{-1,0,1\},
\end{equation}
with randomly sampled coefficient tensors $\alpha_{x_1x_2x_3}$. For each functional, its classical bound is obtained by maximizing over all deterministic vertices. Below, we provide one of them that exhibits quantum advantage:
\begin{equation}
\begin{aligned}
\mathcal{I}_{12}&= -p_{000} + p_{002} - p_{010} + p_{011} + p_{100} 
   - p_{101} - p_{102} \nonumber \\
  & + p_{110} + p_{111} 
   - p_{112} + p_{120} + p_{121} - p_{122} + p_{202} \nonumber \\
   & + p_{210} 
   + p_{212} - p_{220} - p_{221} \leqslant 8.
\end{aligned}
\label{eq:four-party-ineq}
\end{equation}
A quantum advantage for $\mathcal{I}_{12} = 7+\sqrt{2} \cong 8.414$ is obtained in the $S_1\rightarrow S_3,\;S_2\rightarrow S_3,\;S_3\rightarrow R$ topology, using the SDP seesaw procedure. The explicit quantum strategy is given in Appendix~\ref{app:Q-strategies}.

\section{Conclusions}

We have introduced topology-independent quantum advantage, where quantum communication outperforms every classical protocol satisfying the same communication constraint, irrespective of how the classical communication resources are distributed among the parties. We demonstrated this phenomenon in a minimal three-party scenario, characterized the corresponding topology-independent classical correlation polytope, and identified several quantum violations. Remarkably, one of the resulting inequalities also enables semi-device-independent certifications of nonunitary quantum channels and quantum network topologies. We further extended the framework to a four-party network and found additional topology-independent quantum advantages.

Our results open several directions for future work. An important question is whether one can construct families of communication tasks exhibiting increasingly large, or even unbounded, separations between quantum and classical strategies while remaining independent of the communication topology. The robustness of the advantage against arbitrary classical network architectures also motivates the search for cryptographic and semi-device-independent applications whose security or performance does not require assumptions about the underlying communication topology.

It would also be interesting to extend the present framework beyond a single input-free output party. Natural generalizations include allowing inputs at the output party, multiple parties producing outputs, or more general scenarios in which several parties both receive inputs and generate outputs. Understanding topology-independent quantum advantages in larger networks, and determining whether the associated communication inequalities can certify more general properties of quantum channels and operations, constitute further promising directions.

\textit{Note added.} AP, VNSMA, and LGP contributed equally. The use of generative artificial intelligence tools was limited to language editing. The numerical codes used for the facet enumeration and for obtaining the quantum advantages presented in this work are publicly available in Ref.~\cite{codes}.

\bibliography{ref} 

\appendix

\section{Proof of Theorems~\ref{lem:1}-~\ref{lem:four-party-reduction} }\label{App:proofs}

\begin{proof}[Proof of Theorem~\ref{lem:1}]
We first consider classical communication. For a fixed value of the
shared random variable $\lambda$, a general strategy in the parallel
topology has the form
\begin{equation}
p(z|x,y,\lambda)
=
\sum_{m_A,m_B}
p(m_A|x,\lambda)
p(m_B|y,\lambda)
p(z|m_A,m_B,\lambda),
\end{equation}
where $m_A\in[d_A]$ and $m_B\in[d_B]$.

Consider now the topology $A\rightarrow B\rightarrow C$. Alice sends
the same message $m_A$ to Bob. Since Bob knows both $m_A$ and his local
input $y$, he can directly implement the effective response function
\begin{equation}
\widetilde p(z|m_A,y,\lambda)
=
\sum_{m_B}
p(m_B|y,\lambda)
p(z|m_A,m_B,\lambda).
\end{equation}
Bob samples $z$ according to this distribution and sends the resulting
value to Charlie. The second communication therefore requires only a
message of dimension $n_z$. Charlie simply outputs the received value.
The resulting correlation is
\begin{equation}
\begin{split}
p(z|x,y,\lambda)
&=
\sum_{m_A}
p(m_A|x,\lambda)
\widetilde p(z|m_A,y,\lambda)\\
&=
\sum_{m_A,m_B}
p(m_A|x,\lambda)
p(m_B|y,\lambda)
p(z|m_A,m_B,\lambda),
\end{split}
\end{equation}
which coincides exactly with the correlation generated by the parallel
topology. Averaging over $\lambda$ preserves the equality.

The quantum case follows from the same observation. Let Alice and Bob
send states $\rho_x$ and $\sigma_y$, of dimensions $d_A$ and $d_B$,
respectively, and let $\{M_z\}_{z\in[n_z]}$ be Charlie's POVM. The
parallel strategy produces
\begin{equation}
p(z|x,y)
=
\operatorname{Tr}
\left[
M_z(\rho_x\otimes\sigma_y)
\right].
\end{equation}
For every fixed $y$, define operators acting on Alice's communicated
system by
\begin{equation}
E_{z|y}
=
\operatorname{Tr}_{B}
\left[
(\mathbb{I}\otimes\sigma_y^{1/2})
M_z
(\mathbb{I}\otimes\sigma_y^{1/2})
\right].
\end{equation}
These operators form a valid POVM, since $E_{z|y}\geq0$,
and $\sum_z E_{z|y}
= \operatorname{Tr}_{B}
\left[
\mathbb{I}\otimes\sigma_y
\right] = \mathbb{I}.$
Moreover,
\begin{equation}
\operatorname{Tr}
\left[
E_{z|y}\rho_x
\right]
=
\operatorname{Tr}
\left[
M_z(\rho_x\otimes\sigma_y)
\right].
\end{equation}

Hence, in the sequential topology $A\rightarrow B\rightarrow C$,
Alice sends $\rho_x$ to Bob, who performs the POVM
$\{E_{z|y}\}_z$ determined by his input $y$. Bob then communicates the
measurement outcome $z$ to Charlie, who outputs it directly. The first
message has dimension $d_A$, while the second has dimension $n_z$, and
the resulting correlation is identical to that of the original
parallel strategy.

The simulation in the topology $B\rightarrow A\rightarrow C$ follows
by exchanging Alice and Bob.
\end{proof}

\begin{proof}[Proof of Theorem~\ref{lem:2}]
As an immediate consequence, under the additive communication-dimension
constraint of \textit{Theorem}~\ref{lem:1}, the topology $A\rightarrow C,\;B\rightarrow C$ is redundant
whenever
\begin{equation}
n_z\leq\max\{d_A,d_B\}.
\end{equation}
Indeed, if $n_z\leq d_B$, the simulation
$A\rightarrow B\rightarrow C$ uses total dimension
$d_A+n_z\leq d_A+d_B$; if $n_z\leq d_A$, the analogous simulation $B\rightarrow A\rightarrow C$ satisfies the same bound. Since in $[(3,3),2,4]$, the only possibility is $d_A=d_B=2$ and $n_z$ is also at least 2, the parallel topology is redundant. 

Consider the two topologies in which Alice (Bob) communicates with both Bob (Alice) and Charlie, while Bob (Alice) also communicates with Charlie. Each such topology contains three communication channels and is therefore incompatible with the communication-dimension constraint $d=4$, since every nontrivial channel has dimension at least 2. 

Consequently, only the four topologies considered above remain, which completes the proof.
\end{proof}

\begin{proof}[Proof of Theorem~\ref{lem:parallelQ}]
For the parallel topology the correlations have the form
\begin{equation}
    p_{xy}
    =
    \Tr\!\left[
        M_0\left(\rho_x\otimes\sigma_y\right)
    \right],
    \label{eq:parallel-correlation}
\end{equation}
where $\rho_x$ and $\sigma_y$ are qubit states independently prepared by
Alice and Bob, respectively, and $0\leq M_0\leq\mathbb{I}$ is Charlie's
binary measurement effect. Since the objective is linear, it is sufficient
to consider pure preparations. We write
\begin{equation}
    \rho_x=\frac{1}{2}
    \left(\mathbb{I}+\mathbf{a}_x\cdot\boldsymbol{\sigma}\right),
    \quad
    \sigma_y=\frac{1}{2}
    \left(\mathbb{I}+\mathbf{b}_y\cdot\boldsymbol{\sigma}\right),
\end{equation}
with $|\mathbf{a}_x|=|\mathbf{b}_y|=1$. Let
\begin{equation}
    C_5=
    \begin{pmatrix}
        0 & -1 & 1\\
        -1 & 1 & 0\\
        1 & 0 & -1
    \end{pmatrix}
\end{equation}
denote the coefficient matrix of Eq.~\eqref{eq:I5-topology}, and define the quantum operator associated with Eq.~\eqref{eq:I5-topology}
\begin{equation}
    W_5=
    \sum_{x,y=0}^{2}(C_5)_{xy}\,
    \rho_x\otimes\sigma_y .
\end{equation}
Since every row and every column of $C_5$ sums to zero, all identity and
single-party Bloch terms cancel. Hence
\begin{equation}
    W_5=
    \frac{1}{4}
    \sum_{i,j=1}^{3}
    T_{ij}\,
\sigma_i\otimes\sigma_j,
    \quad
    T=A^{\mathsf T}C_5B,
    \label{eq:W5-T}
\end{equation}
where the rows of $A$ and $B$ are the Bloch vectors
$\mathbf{a}_x^{\mathsf T}$ and $\mathbf{b}_y^{\mathsf T}$, respectively.

Since $\rank(C_5)=2$, Eq.~\eqref{eq:W5-T} implies
$\rank(T)\leq2$. Let $s_1\geq s_2\geq0$ denote the two
nonzero singular values of $T$. A signed singular-value decomposition
can be written as
\begin{equation}
    T=O_A
    \operatorname{diag}(t_1,t_2,0)
    O_B^{\mathsf T},
    \quad
    |t_1|=s_1,\ \ |t_2|=s_2,
\end{equation}
where $O_A,O_B\in SO(3)$. For every
$O\in SO(3)$ there exists a qubit unitary $U$ satisfying
\begin{equation}
    U\sigma_iU^\dagger
    =
    \sum_{j=1}^3 O_{ji}\sigma_j.
\end{equation}
Hence, choosing local unitaries $U_A$ and $U_B$ associated with
$O_A$ and $O_B$, respectively, transforms the correlation matrix as
$ O_ATO_B^{\mathsf T}=\operatorname{diag}(t_1,t_2,0)$.
Consequently,
\bea
    W'_5
    &:=&
    (U_A\otimes U_B)W_5
    (U_A^\dagger\otimes U_B^\dagger) \nonumber \\
   & =&
    \frac{1}{4}
    \left(
        t_1\sigma_1\otimes\sigma_1
        +
        t_2\sigma_2\otimes\sigma_2
    \right).
    \label{eq:W5-canonical}
\eea 
Thus $W_5$ and $W'_5$ have the same spectrum. Since
$\sigma_1\otimes\sigma_1$ and $\sigma_2\otimes\sigma_2$ commute,
the eigenvalues of $W_5$ are
\begin{equation}
    \frac{1}{4}
    \left\{
        t_1+t_2,\,
        t_1-t_2,\,
        -t_1+t_2,\,
        -t_1-t_2
    \right\}.
\end{equation}
Using $|t_1|=s_1\geq |t_2|=s_2$, the sum of the positive
eigenvalues is therefore
$\frac{1}{4}\left[(s_1+s_2)+(s_1-s_2)\right]=\frac{s_1}{2}.$

Charlie's binary measurement effect is only required to satisfy
$0\leq M_0\leq\mathbb{I}$; in particular, no restriction on
$\Tr(M_0)$ is imposed. For any Hermitian operator $W$ and any effect
$0\leq M\leq\mathbb{I}$,
\begin{equation}
    \Tr(MW)\leq\Tr(W_+),
\end{equation}
where $W_+$ denotes the positive part of $W$. The maximum is
attained by choosing $M$ to be the projector onto the positive
eigenspace of $W$. Consequently,
\begin{equation}
    \mathcal{I}_5
    =
    \Tr(M_0W_5)
    \leq
    \Tr[(W_5)_+]
    =
    \frac{s_1}{2}.
    \label{eq:I5-positive}
\end{equation}
For arbitrary real unit vectors $\mathbf{u}$ and $\mathbf{v}$,
\begin{equation}
    \mathbf{u}^{\mathsf T}T\mathbf{v}
    =
    \boldsymbol{\alpha}^{\mathsf T}
    C_5
    \boldsymbol{\beta},
\end{equation}
where
$\alpha_x=\mathbf{a}_x\cdot\mathbf{u},\beta_y=\mathbf{b}_y\cdot\mathbf{v},$ and consequently $|\alpha_x|,|\beta_y|\leq1$. It follows that
\begin{equation}
    s_1=\|T\|_{\infty}
    \leq
    \max_{\boldsymbol{\alpha},\boldsymbol{\beta}\in[-1,1]^3}
    \boldsymbol{\alpha}^{\mathsf T}
    C_5
    \boldsymbol{\beta}
    =4,
\end{equation}
where the last maximum is attained at extremal points
$\alpha_x,\beta_y\in\{-1,1\}$. Substitution into
Eq.~\eqref{eq:I5-positive} gives
$I_{5,A\rightarrow C \leftarrow B}^{Q}\leq2.$

For the topology $A\rightarrow C$, the correlations are independent of
Bob's input $y$, i.e.,
$p_{xy}=p_x .$
Consequently,
\begin{equation}
    I_5
    =
    \sum_x
    \left(\sum_y(C_5)_{xy}\right)p_x
    =0,
\end{equation}
because every row of $C_5$ sums to zero. This conclusion is independent
of the dimension of the communicated quantum system and therefore also
holds when Alice communicates a four-dimensional system. Similarly, for
$B\rightarrow C$, $p_{xy}=p_y$, and the vanishing column sums of $C_5$
give $\mathcal{I}_5=0$. This proves Eq.~\eqref{eq:I5-topology}.
\end{proof}

\begin{proof}[Proof of Theorem~\ref{lem:four-party-reduction}] Since every nontrivial communication channel has dimension at least $2$ and the total communication dimension is bounded by $d=6$, an admissible topology contains at most three edges. We classify the possibilities according to the number of edges and repeatedly use Theorem~\ref{lem:1} to eliminate redundant parallel communication to the input-free receiver $R$. 

For a single edge, the only possible structure is $ S_i\rightarrow R$. There are three choices of the communicating sender, giving three topologies shown in Fig.~\ref{fig4:1}. 

For two edges, there are two graph structures. The parallel topology $S_i\rightarrow R,\ S_j\rightarrow R$ is redundant by Theorem~\ref{lem:1}: its correlations can be reproduced by a sequential topology in which one sender communicates with the other, who performs the effective decoding and communicates the binary outcome to $R$. It therefore suffices to retain $S_i\rightarrow S_j\rightarrow R$. There are $3\times2=6$ ordered choices of $(S_i,S_j)$, yielding six topologies given by Fig.~\ref{fig4:2}. Moreover, since $R$ has a binary output, the final communication can always be chosen to have dimension $2$; the remaining communication resource can therefore be assigned to the first edge. 

For three edges, all messages necessarily have dimension $2$. Any topology containing two independent edges terminating at $R$ can again be reduced using Theorem~\ref{lem:1}. For example, a topology containing $S_i\rightarrow S_j,\ S_j\rightarrow R,\ S_k\rightarrow R$ can be replaced by $S_i\rightarrow S_j\rightarrow S_k\rightarrow R,$ where $S_k$ performs the effective decoding associated with the two messages that would otherwise arrive independently at $R$. Thus such graphs reduce to the fully sequential topology (Fig.~\ref{fig4:3a}). Similarly, the fully parallel topology $S_i\rightarrow R,\ S_j\rightarrow R,\ S_k\rightarrow R$ can be reduced to a topology in which two senders communicate with the third sender, who performs the effective decoding and sends the binary result to $R$. This gives the converging topology $S_j\rightarrow S_i,\ S_k\rightarrow S_i,\ S_i\rightarrow R$ (Fig.~\ref{fig4:3b}). The remaining three-edge DAGs are either equivalent, under relabeling of the senders, to one of these two structures or reduce to one of the two-edge sequential topologies by the same argument. Consequently, only the fully sequential and converging structures need to be retained among three-edge graphs. The fully sequential topology admits all $3!=6$ permutations of the senders. For the converging topology, only the choice of the terminal sender $S_k$ is relevant, since exchanging the two incoming senders leaves the graph unchanged; hence there are three such topologies. Altogether, it is therefore sufficient to consider $3+6+6+3=18$ topologies. 
\end{proof}

\section{Quantum advantageous strategies}\label{app:Q-strategies}

\subsection{$\mathcal{I}_5$ from Table~\ref{TABLE1}, in $[(3,3),2,4]$ scenario}

The encoding strategy by Alice is given as $\rho_{x_1} = \ket{\psi_{x_1}}\bra{\psi_{x_1}}$:
\bea
&\ket{\psi_{0}} = \ket{0},  \ket{\psi_{1}} = \frac{1}{2} \Bigl(\ket{0} + \sqrt{3} \ket{1} \Bigr), \nonumber \\
&\ket{\psi_2} = \frac{1}{2} \Bigl(\ket{0} - \sqrt{3} \ket{1} \Bigr) 
\eea
Bob's channel is specified by two Kraus operators,
$\mathcal{E}_{x_2}(\rho)=K_{x_2}\rho K_{x_2}^\dagger+K_{x_2}'\rho K_{x_2}'^\dagger$.
For $x_2=0$,
\begin{equation}
K_0 = \begin{pmatrix}
    \frac{(1+i)}{2\sqrt{2}} & \frac{(-1+i)}{2\sqrt{2}} \\
    - \frac{i}{2} & -\frac{1}{2}
\end{pmatrix}, \quad 
K_0' = \begin{pmatrix}
    \frac{(1+i)}{2\sqrt{2}} & \frac{(1-i)}{2\sqrt{2}} \\
     \frac{i}{2} & -\frac{1}{2}
\end{pmatrix}.
\end{equation}
For $x_2=1$,
\begin{equation}
K_1 = \begin{pmatrix}
    - \frac{\sqrt{3}}{2\sqrt{2}} & \frac{1}{2\sqrt{2}} \\
     \frac{1-i}{2} &  0
\end{pmatrix}, 
K_1' = \begin{pmatrix}
   \frac{(1+i)}{4} & \frac{\sqrt{3}(1+i)}{4} \\
     0 & -\frac{1}{\sqrt{2}}
\end{pmatrix}
\end{equation}
For $x_2=2$,
\begin{equation}
K_2 = \begin{pmatrix}
    \frac{\sqrt{3}}{2\sqrt{2}} & \frac{1}{2\sqrt{2}} \\
     \frac{(1-i)}{2} &  0
\end{pmatrix},
K_2' = \begin{pmatrix}
    \frac{(1+i)}{4} & \frac{-\sqrt{3}(1+i)}{4} \\
     0 &  \frac{1}{\sqrt{2}}
\end{pmatrix}.
\end{equation}
Charlie's measurement is given as $\{\ket{\phi}\!\bra{\phi},\mathbb{I} - \ket{\phi}\!\bra{\phi} \}$. Where $\ket{\phi}$ is given by
\bea
\ket{\phi} = \frac{1}{2}(\ket{0} + e^{-\frac{i \pi}{4}}\ket{1} ).
\eea
Using these states, transformations and measurements, we get the quantum value for $\mathcal{I}_5$ to be $\frac{3\sqrt{3}}{2} =  2.5981$.

\subsection{$\mathcal{I}_8$ from Table~\ref{TABLE1}, in $[(3,3),2,4]$ scenario}

Let us first define $\theta = \frac{\pi}{8}$. The encoding strategy by Alice is given by
\bea
& \ket{\psi_0} = \ket{0},\ 
\ket{\psi_1} = \frac{1}{\sqrt{2}}
\left(\ket{0}+\ket{1}\right), \nonumber \\
& \ket{\psi_2} = \sin(\theta) \ket{0} - \cos(\theta) \ket{1}  .
\eea
Similar as before each channel is specified by two Kraus operators,
$\mathcal{E}_{x_2}(\rho)=K_{x_2}\rho K_{x_2}^\dagger+K_{x_2}'\rho K_{x_2}'^\dagger$.
For $x_2=0$,
\begin{equation}
K_0=
\frac{1}{\sqrt{2}}
\begin{pmatrix}
\sin\theta & \cos\theta\\
-\cos\theta & \sin\theta
\end{pmatrix},
\ 
K'_0=
\frac{1}{\sqrt{2}}
\begin{pmatrix}
-\cos\theta & \sin\theta\\
\sin\theta & \cos\theta
\end{pmatrix};
\end{equation}
for $x_2=1$,
\begin{equation}
K_1 = \frac{1}{\sqrt{2}}\begin{pmatrix}
    -\sin\theta & \cos\theta \\
    \cos\theta & \sin\theta
\end{pmatrix}, \ 
K_1' = \frac{1}{\sqrt{2}}\begin{pmatrix}
    \cos\theta & \sin\theta \\
    -\sin\theta & \cos\theta
\end{pmatrix};
\end{equation}
and for $x_2=2$,
\begin{equation}
K_2 = \frac{1}{\sqrt{2}}\begin{pmatrix}
    0 & -1 \\
    0 & 1
\end{pmatrix}, \
K_2' = \frac{1}{\sqrt{2}}\begin{pmatrix}
    -1 & 0 \\
     1 & 0
\end{pmatrix}.
\end{equation} Charlie's measurement is given as $\{\ket{\phi}\!\bra{\phi},\mathbb{I} - \ket{\phi}\!\bra{\phi} \}$, where $\ket{\phi}=\frac{1}{\sqrt{2}}(\ket{0}-\ket{1})$.
This quantum strategy yields the value $3+\sqrt{2} = 4.414$ for $\mathcal{I}_8$.

\subsection{$\mathcal{I}_{12}$ inequality in $[(3,3,3),2,6]$ scenario}

The encoding strategy of $S_1$ is
$\rho^{(1)}_{x_1}=\ket{\psi^{(1)}_{x_1}}\bra{\psi^{(1)}_{x_1}}$,
where
\begin{equation}
\ket{\psi^{(1)}_0}=\ket{0},\ 
\ket{\psi^{(1)}_1}=\ket{1},\ 
\ket{\psi^{(1)}_2}=\ket{0}.
\end{equation}
The encoding strategy of $S_2$ is is
$\rho^{(2)}_{x_2}=\ket{\psi^{(2)}_{x_2}}\bra{\psi^{(2)}_{x_2}}$,
where
\begin{align}
\ket{\psi^{(2)}_0}
&=
\frac{\ket{0}+\ket{1}}{\sqrt{2}}, \nonumber \\
\ket{\psi^{(2)}_1}
&=
\sin\left(\frac{\pi}{8}\right)\ket{0}
-
\cos\left(\frac{\pi}{8}\right)\ket{1}, \nonumber  \\
\ket{\psi^{(2)}_2}
&=
\cos\left(\frac{\pi}{8}\right)\ket{0}
-
\sin\left(\frac{\pi}{8}\right)\ket{1}.
\end{align}

Upon receiving the two qubits from $S_1$ and $S_2$, the party $S_3$
applies the input-dependent channel
\begin{equation}
\mathcal{E}_{x_3}(\rho)
=
\sum_{j=1}^{6}
K_{x_3}^{j}\rho (K_{x_3}^{j})^\dagger .
\end{equation}
The Kraus operators are as follows.

For $x_3=0$,
\begin{align}
K_0^1&=
\begin{pmatrix}
0&0&0&0\\
0&0&-\frac12&0\\
0&0&0&0\\
0&0&\frac12&0
\end{pmatrix},
&
K_0^2&=
\begin{pmatrix}
0&0&0&\frac12\\
0&0&0&0\\
0&0&0&\frac12\\
0&0&0&0
\end{pmatrix},
\\[1mm]
K_0^3&=
\begin{pmatrix}
0&0&\frac12&0\\
0&0&0&0\\
0&0&\frac12&0\\
0&0&0&0
\end{pmatrix},
&
K_0^4&=
\begin{pmatrix}
0&0&0&0\\
0&0&0&-\frac12\\
0&0&0&0\\
0&0&0&\frac12
\end{pmatrix},
\\[1mm]
K_0^5&=
\begin{pmatrix}
0&0&0&0\\
\frac{1}{\sqrt2}&0&0&0\\
0&0&0&0\\
\frac{1}{\sqrt2}&0&0&0
\end{pmatrix},
&
K_0^6&=
\begin{pmatrix}
0&0&0&0\\
0&\frac{1}{\sqrt2}&0&0\\
0&0&0&0\\
0&\frac{1}{\sqrt2}&0&0
\end{pmatrix}.
\end{align}
For $x_3=1$,
\begin{align}
K_1^1&=
\begin{pmatrix}
0&0&0&0\\
0&-\frac12&0&0\\
0&0&0&0\\
0&\frac12&0&0
\end{pmatrix},
&
K_1^2&=
\begin{pmatrix}
0&\frac12&0&0\\
0&0&0&0\\
0&\frac12&0&0\\
0&0&0&0
\end{pmatrix},
\\[1mm]
K_1^3&=
\frac{1}{2\sqrt2}
\begin{pmatrix}
0&0&0&0\\
0&0&1&-1\\
0&0&0&0\\
0&0&-1&1
\end{pmatrix},
&
K_1^4&=
\frac{1}{2\sqrt2}
\begin{pmatrix}
0&0&1&-1\\
0&0&0&0\\
0&0&1&-1\\
0&0&0&0
\end{pmatrix},
\\[1mm]
K_1^5&=
\frac12
\begin{pmatrix}
0&0&0&0\\
0&0&1&1\\
0&0&0&0\\
0&0&1&1
\end{pmatrix},
&
K_1^6&=
\begin{pmatrix}
0&0&0&0\\
\frac{1}{\sqrt2}&0&0&0\\
0&0&0&0\\
\frac{1}{\sqrt2}&0&0&0
\end{pmatrix}.
\end{align}
For $x_3=2$,
\begin{align}
K_2^1&=
\begin{pmatrix}
0&0&0&0\\
0&-\frac12&0&0\\
0&0&0&0\\
0&\frac12&0&0
\end{pmatrix},
&
K_2^2&=
\begin{pmatrix}
\frac12&0&0&0\\
0&0&0&0\\
\frac12&0&0&0\\
0&0&0&0
\end{pmatrix},
\\[1mm]
K_2^3&=
\begin{pmatrix}
0&\frac12&0&0\\
0&0&0&0\\
0&\frac12&0&0\\
0&0&0&0
\end{pmatrix},
&
K_2^4&=
\begin{pmatrix}
0&0&0&0\\
-\frac12&0&0&0\\
0&0&0&0\\
\frac12&0&0&0
\end{pmatrix},
\\[1mm]
K_2^5&=
\begin{pmatrix}
0&0&0&0\\
0&0&0&\frac{1}{\sqrt2}\\
0&0&0&0\\
0&0&0&\frac{1}{\sqrt2}
\end{pmatrix},
&
K_2^6&=
\begin{pmatrix}
0&0&0&0\\
0&0&\frac{1}{\sqrt2}&0\\
0&0&0&0\\
0&0&\frac{1}{\sqrt2}&0
\end{pmatrix}.
\end{align}
Finally, the receiver performs the binary measurement
$\{M_0,\mathbb{I}_4-M_0\}$, with
\begin{equation}
M_0=
\begin{pmatrix}
1&0&0&0\\
0&\frac{1}{2}&0&-\frac{1}{2}\\
0&0&1&0\\
0&-\frac{1}{2}&0&\frac{1}{2}
\end{pmatrix}.
\end{equation}
For this strategy,
\begin{equation}
p_{x_1x_2x_3}
=
\Tr\!\left[
M_0\,
\mathcal{E}_{x_3}
\left(
\rho^{(1)}_{x_1}\otimes\rho^{(2)}_{x_2}
\right)
\right],
\end{equation}
and substitution into $\mathcal{I}_{12}$ gives
$7+\sqrt{2} \simeq 8.4142>8.$

\end{document}